\documentclass[11pt]{article}
\usepackage[utf8]{inputenc}
\usepackage[T1]{fontenc}
\usepackage{latexsym}
\usepackage{amsmath,amssymb, amsfonts, amsthm, mathtools} 
\usepackage{graphicx}
\usepackage{colortbl}
\usepackage{enumerate}
\usepackage{tikz}
\usepackage{empheq}
\usepackage{ifthen}
\usetikzlibrary{arrows,calc,decorations.pathmorphing}
\usetikzlibrary{decorations.pathreplacing,shapes.misc}
\usepackage[pdfpagelabels=true]{hyperref}
\usepackage{geometry}

\usepackage[square,comma]{natbib} % this gives author [year] or [author, year]
\def\Iscr{\mathcal{I}}

\definecolor{darkgreen}{rgb}{0,0.5,0}

\newcommand{\opt}{\mbox{\scriptsize\rm OPT}}
\newcommand{\lp}{\mbox{\scriptsize\rm LP}}

\newtheorem{theorem}{Theorem}
\newtheorem{lemma}[theorem]{Lemma}
\newtheorem{corollary}[theorem]{Corollary}
\newtheorem{proposition}[theorem]{Proposition}
\newtheorem{definition}[theorem]{Definition}
\def\prove{\par \noindent \hbox{\bf Proof:}\quad}
\def\endproof{\eol \rightline{$\Box$} \par}
\renewcommand{\endproof}{\hspace*{\fill} {\boldmath $\Box$} \par \vskip0.5em}

\renewcommand{\epsilon}{\varepsilon}

\newcommand{\sfrac}[2]{{\textstyle\frac{#1}{#2}}} 

\newcommand{\rhopathNW}{\rho_{\protect\textnormal{ATSPP}}^{\text{NW}}}
\newcommand{\rhoatspNW}{\rho_{\protect\textnormal{ATSP}}^{\text{NW}}}
\newcommand{\rhopath}{\rho_{\protect\textnormal{ATSPP}}}
\newcommand{\rhoatsp}{\rho_{\protect\textnormal{ATSP}}}

\begin{document}

\title{The asymmetric traveling salesman path LP \\ has constant integrality ratio}
\author{Anna K\"ohne \and Vera Traub \and Jens Vygen}
\date{\small Research Institute for Discrete Mathematics, University of Bonn \\
\texttt{\{koehne,traub,vygen\}@or.uni-bonn.de}}

\begingroup
\makeatletter
\let\@fnsymbol\@arabic
\maketitle
\endgroup

\begin{abstract}
We show that the classical LP relaxation of the asymmetric traveling salesman path problem (ATSPP) has constant integrality ratio.
If $\rhoatsp$ and $\rhopath$ denote the integrality ratios for the asymmetric TSP and its path version, then
$\rhopath\le 4\rhoatsp-3$.

We prove an even better bound for node-weighted instances:
if the integrality ratio for ATSP on node-weighted instances is $\rhoatsp^{\text NW}$,
then the integrality ratio for ATSPP on node-weighted instances is at most $2\rhoatsp^{\text NW}-1$.

Moreover, we show that for ATSP node-weighted instances and unweighted digraph instances are almost equivalent.
From this we deduce a lower bound of $2$ on the integrality ratio of unweighted digraph instances.

\end{abstract}

\section{Introduction}

In the asymmetric traveling salesman path problem (ATSPP), we are given a directed graph $G=(V,E)$, two vertices $s,t\in V$,
and weights $c:E\to\mathbb{R}_{\ge 0}\cup\{\infty\}$. We look for a sequence $s=v_0,v_1,\ldots,v_k=t$ that contains
every vertex at least once (an $s$-$t$-tour); the goal is to minimize $\sum_{i=1}^k c(v_{i-1},v_i)$. 
Equivalently, we can assume that $G$ is complete and the triangle
inequality $c(u,v)+c(v,w)\ge c(u,w)$ holds for all $u,v,w\in V$, and require the sequence to contain every vertex exactly once.

The special case $s=t$ is known as the asymmetric traveling salesman problem (ATSP).
In a recent breakthrough, \cite{SveTV18} found the first constant-factor approximation algorithm for ATSP,
and they also proved that its standard LP relaxation has constant integrality ratio.

\cite{FeiS07} showed that any $\alpha$-approximation algorithm for ATSP implies a $(2\alpha+\epsilon)$-approximation algorithm
for ATSPP (for any $\epsilon>0$).
Hence ATSPP also has a constant-factor approximation algorithm.
In this paper we prove a similar relation for the integrality ratios.
This answers an open question by \cite{FriGS16}.

Given that the upper bound on the integrality ratio by \cite{SveTV18} is a large constant that will probably be improved in the future, such a 
blackbox result seems particulary desirable. 
Any improved upper bound on the integrality ratio for ATSP then immediately implies a better bound for the path version.

\subsection{The linear programming relaxation}

The classical linear programming relaxation for ATSPP (for $s\ne t$) is
 \begin{equation}\label{eq:subtour_lp_path}
  \begin{aligned}
    \min c(x)  \\
   s.t. & &  x(\delta^-(s)) -x(\delta^+(s)) =& -\! 1 & \\
   & &  x(\delta^-(t)) -x(\delta^+(t)) =&\ 1 & \\
   & &  x(\delta^-(v)) -x(\delta^+(v)) =&\ 0 & & \text{ for } v\in V\setminus\{s,t\} \\
   & &  x(\delta(U)) \ge&\ 2 & &\text{ for } \emptyset \ne U \subseteq V\setminus\{s,t\} \\
   & & x_e \ge & \ 0 & &\text{ for } e\in E
  \end{aligned}
  \tag{ATSPP LP}
 \end{equation}
Here (and henceforth) we write $c(x):=\sum_{e\in E} c(e) x_{e}$, $x(F):=\sum_{e\in F}x_{e}$, 
$\delta^+(U):=\{(u,v)\in E: u\in U,\, v\in V\setminus U\}$, $\delta^-(U):=\delta^+(V\setminus U)$, $\delta(U):=\delta^-(U)\cup\delta^+(U)$,
$\delta^+(v):=\delta^+(\{v\})$, and $\delta^-(v):=\delta^-(\{v\})$.
For an instance $\Iscr$ we denote by $\lp_{\Iscr}$ the value of an optimum solution to \eqref{eq:subtour_lp_path} and by $\opt_{\Iscr}$
the value of an optimum integral solution.
If the instance is clear from the context, we will sometimes simply write $\lp$ and $\opt$.
Note that the integral solutions of \eqref{eq:subtour_lp_path} are precisely the incidence vectors of multi-digraphs $(V,F)$ 
that become connected and Eulerian by adding one edge $(t,s)$.
Hence they correspond to walks from $s$ to $t$ that visit all vertices, in other words: $s$-$t$-tours.

The integrality ratio of \eqref{eq:subtour_lp_path}, denoted by $\rhopath$, is the maximal ratio of an optimum integral solution and an optimum fractional solution;
more precisely $\sup_{\Iscr} \frac{\opt_{\Iscr}}{\lp_{\Iscr}}$, 
where the supremum goes over all instances $\Iscr=(G,c,s,t)$ with $s\ne t$ for which
the denominator is nonzero and finite.
\cite{NagR08} proved that $\rhopath=O(\sqrt{n})$, where $n=|V|$.
This bound was improved to $O(\log n)$ by \cite{FriSS13} and to $O(\log n/  \log\log n)$ by \cite{FriGS16}.
In this paper we prove that the integrality ratio of \eqref{eq:subtour_lp_path} is in fact constant.

Let $\rhoatsp$ denote the integrality ratio of the classical linear programming relaxation for ATSP:
 \begin{equation}\label{eq:subtour_lp_atsp}
  \begin{aligned}
   \min c(x)  \\
   s.t.& &  x(\delta^-(v)) -x(\delta^+(v)) =&\ 0 & & \text{ for } v\in V \\
   & &  x(\delta(U)) \ge&\ 2 & &\text{ for } \emptyset \ne U \subsetneq V \\
   & & x_e \ge & \ 0 & &\text{ for } e\in E
  \end{aligned}
    \tag{ATSP LP}
 \end{equation}
 
\cite{SveTV18} proved that $\rhoatsp$ is a constant.
By an infinite sequence of instances, \cite{ChaGK06} showed that $\rhoatsp\ge 2$.
It is obvious that $\rhopath\ge\rhoatsp$: split an arbitrary vertex of an ATSP instance
into two copies, one (called $s$) inheriting the outgoing edges, and one (called $t$) inheriting the entering edges; add an edge $(t,s)$ of 
cost zero and with $x_{(t,s)}:=x(\delta^+(s))-1$.
Figure \ref{fig:example} displays a simpler family of examples, due to \cite{FriGS16}, showing that $\rhopath\ge 2$.

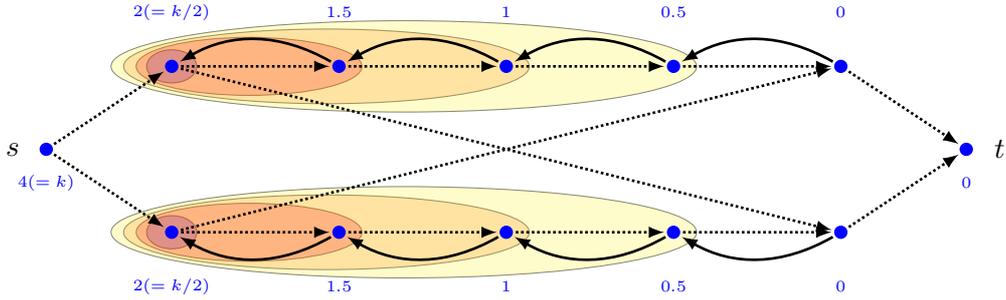
\begin{figure}[ht]
\begin{center}
\begin{tikzpicture}[scale=1.1]
\tikzstyle{vertex}=[blue,circle,fill,minimum size=5,inner sep=0,outer sep=1]
\tikzstyle{expensiveedge}=[->, >=latex, line width=1]
\tikzstyle{cheapedge}=[->,densely dotted,>=latex, line width=1]
\tikzstyle{a}=[blue]
\def\opac{0.2}
\def\opacline{0.5}

\draw[opacity=\opacline](4.275,2) ellipse ({3.5} and {0.55});
\draw[opacity=\opacline](3.35,2) ellipse ({2.425} and {0.45});
\draw[opacity=\opacline](2.425,2) ellipse ({1.35} and {0.35});
\draw[opacity=\opacline](1.5,2) ellipse ({0.3} and {0.2});
\fill[yellow, opacity=\opac] (4.275,2) ellipse ({3.5} and {0.55});
\fill[orange, opacity=\opac] (3.35,2) ellipse ({2.425} and {0.45});
\fill[red, opacity=\opac] (2.425,2) ellipse ({1.35} and {0.35});
\fill[blue!50!purple, opacity=\opac] (1.5,2) ellipse ({0.3} and {0.2});
\draw[opacity=\opacline](4.275,0) ellipse ({3.5} and {0.55});
\draw[opacity=\opacline](3.35,0) ellipse ({2.425} and {0.45});
\draw[opacity=\opacline](2.425,0) ellipse ({1.35} and {0.35});
\draw[opacity=\opacline](1.5,0) ellipse ({0.3} and {0.2});
\fill[yellow, opacity=\opac] (4.275,0) ellipse ({3.5} and {0.55});
\fill[orange, opacity=\opac] (3.35,0) ellipse ({2.425} and {0.45});
\fill[red, opacity=\opac] (2.425,0) ellipse ({1.35} and {0.35});
\fill[blue!50!purple, opacity=\opac] (1.5,0) ellipse ({0.3} and {0.2});

\node at ( -0.4, 1) {$s$};

\node[vertex] (s)  at ( 0, 1) {};

\node[vertex] (a1)  at ( 1.5, 2) {};
\node[vertex] (b1)  at ( 1.5, 0) {};
\node[vertex] (a2)  at ( 3.5, 2) {};
\node[vertex] (b2)  at ( 3.5, 0) {};
\node[vertex] (a3)  at ( 5.5, 2) {};
\node[vertex] (b3)  at ( 5.5, 0) {};
\node[vertex] (a4)  at ( 7.5, 2) {};
\node[vertex] (b4)  at ( 7.5, 0) {};
\node[vertex] (a5)  at ( 9.5, 2) {};
\node[vertex] (b5)  at ( 9.5, 0) {};
\node[vertex] (t)  at ( 11, 1) {};
\node at ( 11.4, 1) {$t$};

\draw[cheapedge] (s) to (a1);
\draw[cheapedge] (s) to (b1);
\draw[cheapedge] (a1) to (a2);
\draw[cheapedge] (a2) to (a3);
\draw[cheapedge] (a3) to (a4);
\draw[cheapedge] (a4) to (a5);
\draw[cheapedge] (a1) to (b5);
\draw[cheapedge] (b1) to (b2);
\draw[cheapedge] (b2) to (b3);
\draw[cheapedge] (b3) to (b4);
\draw[cheapedge] (b4) to (b5);
\draw[cheapedge] (b1) to (a5);
\draw[cheapedge] (a5) to (t);
\draw[cheapedge] (b5) to (t);
\draw[expensiveedge, bend right] (a5) to (a4);
\draw[expensiveedge, bend right] (a4) to (a3);
\draw[expensiveedge, bend right] (a3) to (a2);
\draw[expensiveedge, bend right] (a2) to (a1);
\draw[expensiveedge, bend left] (b5) to (b4);
\draw[expensiveedge, bend left] (b4) to (b3);
\draw[expensiveedge, bend left] (b3) to (b2);
\draw[expensiveedge, bend left] (b2) to (b1);
\begin{tiny}
\node[a] at (0, 0.6) {$4(=k)$};
\node[a] at (1.5,2.65) {$2(=k/2)$};
\node[a] at (1.5,-0.65){$2(=k/2)$};
\node[a] at (3.5,2.65) {$1.5$};
\node[a] at (3.5,-0.65){$1.5$};
\node[a] at (5.5,2.65) {$1$};
\node[a] at (5.5,-0.65){$1$};
\node[a] at (7.5,2.65) {$0.5$};
\node[a] at (7.5,-0.65){$0.5$};
\node[a] at (9.5,2.65) {$0$};
\node[a] at (9.5,-0.65){$0$};
\node[a] at (11, 0.6) {$0$};

\end{tiny}
\end{tikzpicture}
\end{center}
\caption{Example with integrality ratio approaching 2 as the number of vertices increases. 
Setting $x_e:=\frac{1}{2}$ for all shown edges defines a feasible solution of \eqref{eq:subtour_lp_path}. 
If the $2k$ curved edges have cost 1 and the dotted edges have cost 0, we have $\lp=c(x)=k$, but any $s$-$t$-tour costs at least $2k-1$. 
(In the figure, $k=4$.) 
Setting $y_U=\frac{1}{2}$ for the vertex sets indicated by the ellipses and 
$a_v$ as shown in blue defines an optimum solution of \eqref{eq:dual_subtour_lp_path}. \label{fig:example}}
\end{figure}

\subsection{Our results and techniques}

Our main result says that $\rhopath\le 4\rhoatsp-3$.
Together with \cite{SveTV18}, this implies a constant integrality ratio for \eqref{eq:subtour_lp_path}.

Similarly as \cite{FeiS07}, we transform our ATSPP instance to an ATSP instance by
adding a feedback path from $t$ to $s$ and work with an integral solution to this ATSP instance.
This may use the feedback path several times and hence consist of several $s$-$t$-walks in the
original instance. We now merge these to a single $s$-$t$-walk that contains all vertices.
In contrast to \cite{FeiS07}, the merging procedure cannot use an optimum $s$-$t$-tour, but only
an LP solution. 
Our merging procedure is similar to one step of the approximation algorithm for ATSP by 
\cite{SveTV18}, but our analysis is more involved. 
The main difficulty is that the reduction of ATSP to so-called ``laminarly-weighted'' instances used 
by \cite{SveTV18} does not work for the path version. 

In Section \ref{sect:bound_int_ratio}, we describe our merging procedure and obtain a first bound on the cost of our 
single $s$-$t$-walk that contains all vertices. 
However, this bound still depends on the difference of two dual LP variables corresponding to the vertices $s$ and $t$.
In Section \ref{sect:bound_pot_difference} we give a tight upper bound on this value, which will imply our main result 
$\rhopath\le 4\rhoatsp-3$.

The main lemma that we use to prove this bound essentially says that adding an edge $(t,s)$ of cost equal to the LP value 
does not change the value of an optimum LP solution. 
Note that using the new edge $(t,s)$ with value one or more is obviously pointless, 
but it is not clear that this edge will not be used at all.

For node-weighted instances we obtain a better result:
if the integrality ratio for ATSP on node-weighted instances is $\rhoatsp^{\text NW}$,
then the integrality ratio for ATSPP on node-weighted instances is at most $2\rhoatsp^{\text NW}-1$.
\cite{Sve15} showed that $\rhoatsp^{\text NW}\le 13$.

\cite{BoyEM15} gave a family of node-weighted instances that shows $\rhoatsp^{\text NW}\ge 2$.
In Section \ref{sec:nodeweighted=unweighted} we observe that for ATSP
 node-weighted instances behave in the same way as unweighted instances.
Hence for ATSP there is a family of unweighted digraphs whose integrality ratio tends to 2.
Therefore such a family exists also for ATSPP.

\section{Preliminaries}

Given an instance $(G,c,s,t)$ and an optimum solution $x^*$ to \eqref{eq:subtour_lp_path}, we
may assume that $G=(V,E)$ is the support graph of $x^*$; so $ x^*_{e}>0$ for all $e\in E$.
(This is because omitting edges $e$ with $x^*_e=0$ does not change the optimum LP value and can 
only increase the cost of an optimum integral solution.)
We consider the dual LP of \eqref{eq:subtour_lp_path}:
  \begin{equation}\label{eq:dual_subtour_lp_path}
  \begin{aligned}
    \max & &  a_t - a_s + \sum_{\emptyset \ne U \subseteq V\setminus\{s,t\}} \! 2 y_U \\
    s.t. & & a_w - a_v + \sum_{U: e\in\delta(U)} \! y_U \le&\ c(e) & &\text{ for } e=(v,w)\in E \\
   & & y_U \ge&\ 0 & &\text{ for } \emptyset \ne U \subseteq V\setminus\{s,t\}. 
  \end{aligned}
  \tag{ATSPP DUAL}
 \end{equation}
The \emph{support} of $y$ is the set of nonempty subsets $U$ of $V\setminus\{s,t\}$ for which $y_U>0$. We denote it by $\text{supp}(y)$.
We say that a dual solution $(a,y)$ \emph{has laminar support} 
if for any two nonempty sets $A,B\in\text{supp}(y)$ 
%with $y_A>0$ and $y_B>0$ 
we have $A\cap B=\emptyset$, $A\subseteq B$, or $B\subseteq A$.
See Figure \ref{fig:example} for an example.
We recall some well-known properties of primal and dual LP solutions (cf.\ \cite{SveTV18}) and sketch proofs for sake of completeness:

\begin{proposition}
\label{duallaminar}
Let $(a,y)$ be an optimum solution to \eqref{eq:dual_subtour_lp_path}. Then there is a vector $y'$ such that 
$(a, y')$ is an optimum solution to \eqref{eq:dual_subtour_lp_path} and has laminar support.
\end{proposition}

\prove
Among all $y'$ such that $(a,y')$ is an optimum dual solution, choose $y'$ so that $\sum_{U}y'_U|U|$ is minimum. 
Then $(a,y')$ has laminar support: suppose $y'_A>0$ and $y'_B>0$ and $A\cap B,A\setminus B,B\setminus A\not=\emptyset$, 
then we could decrease $y'_A$ and $y'_B$ and increase $y'_{A\setminus B}$ and $y'_{B\setminus A}$
while maintaining dual feasibility.
\endproof

\begin{proposition}
\label{structureoftightsets}
 Let $(G,c,s,t)$ be an instance of ATSPP, where $G$ is the support graph of an optimum solution $x^*$ to \eqref{eq:subtour_lp_path}.
Let $(a,y)$ be an optimum solution of \eqref{eq:dual_subtour_lp_path}.
Let $U\in\{V\}\cup\textnormal{supp}(y)$.
Then the strongly connected components of $G[U]$ can be numbered $U_1,\ldots,U_l$ such that
$\delta^-(U)=\delta^-(U_1)$, $\delta^+(U)=\delta^+(U_l)$, and 
$\delta^+(U_i)=\delta^-(U_{i+1})\ne \emptyset$ for $i=1,\ldots,l-1$.
If $U=V$, then $s\in U_1$ and $t\in U_l$.
\end{proposition}

\prove
By complementary slackness, $y_U>0$ implies $x^*(\delta(U))=2$ and hence $x^*(\delta^-(U))=1$.
We prove the statement of the Proposition for $U=V$ and every set $\emptyset \ne U \subseteq V\setminus\{s\}$ with
$x^*(\delta^-(U))=1$.
Let $U_1,\ldots,U_l$ be a topological order of the strongly connected components of $G[U]$.
If $U=V$, we have $s\in U_1$ and $t\in U_l$ because for every vertex $v$ in $G$, $v$ is reachable from $s$, and $t$ is reachable from $v$.

We now use induction on $l$. For $l=1$, the statement is trivial.
Now assume $l>1$. % and let $U=V$ or $\emptyset \ne U \subseteq V\setminus\{s\}$ with $x^-(\delta(U))=1$.
If $U\ne V$, we have $s\notin U$ and thus $1 \le x^*(\delta^-(U_1)) \le x^*(\delta^-(U)) =1$. 
Thus, in any case (also if $U=V$) $\delta^-(U_1) = \delta^-(U)$.
This implies $\delta^-(U\setminus U_1) \subseteq \delta^+(U_1)$.

If $U=V$, we have $s\in U_1$, $t\notin U_1$, and $x^*(\delta^-(U_1))=0$; therefore
we have $x^*(\delta^+(U_1)) =1$. 
Otherwise, (for $U\ne V$) we have $s\notin U$ and thus $x^*(\delta^+(U_1)) \le x^*(\delta^-(U_1))=1$.
So in both cases $\delta^-(U\setminus U_1) \subseteq \delta^+(U_1)$ and 
$x^*(\delta^+(U_1))\le 1 \le x^*(\delta^-(U\setminus U_1))$.
Thus, we have $\delta^-(U\setminus U_1) =\delta^+(U_1)$ and $x^*(\delta^-(U\setminus U_1))=1$.
Hence, applying the induction hypothesis to $U\setminus U_1$ completes the proof.
\endproof

\begin{proposition}
\label{prop:enterandleaveonlyonce}
 Let $(G,c,s,t)$ be an instance of ATSPP, where $G$ is the support graph of an optimum solution to \eqref{eq:subtour_lp_path}.
Let $(a,y)$ be an optimum solution to \eqref{eq:dual_subtour_lp_path} with laminar support.
Let $\bar U\in\{V\}\cup\textnormal{supp}(y)$ and $v,w\in \bar U$.
If $w$ is reachable from $v$ in the induced subgraph $G[\bar U]$, then there is a $v$-$w$-path in $G[\bar U]$ 
that enters and leaves every set $U\in \text{supp}(y)$ at most once.
\end{proposition}

\prove
Let $P$ be a path from $v$ to $w$ in $G[\bar U]$. Repeat the following. 
Let $U$ be a maximal set with $y_U>0$ that $P$ enters or leaves more than once.
If $P$ enters $U$ more than once, let $v'$ be the vertex after entering the first time and $w'$ the vertex after entering the last time.
By Proposition \ref{structureoftightsets}, $v'$ and $w'$ are in the same strongly connected component of $G[U]$.
We replace the $v'$-$w'$-subpath of $P$ by a path in $G[U]$. Proceed analogously if $P$ leaves $U$ more than once.
\endproof

\section{Bounding the integrality ratio}\label{sect:bound_int_ratio}

We first transform an instance and a solution to \eqref{eq:subtour_lp_path} to an instance and a solution to \eqref{eq:subtour_lp_atsp} 
and work with an integral solution of this ATSP instance.
The following lemma is essentially due to \cite{FeiS07}. For completeness, we prove it here again for our setting.
 \begin{lemma}\label{lemma:bound_by_path}
  Let $d\ge 0$ be a constant. Then $\rhopath \le (d+1)\rhoatsp-d$ if
  the following condition holds for every instance $\Iscr=(G,c,s,t)$ of ATSPP, where $G$ is the
  support graph of an optimum solution to \eqref{eq:subtour_lp_path}:
  If there are $s$-$t$-walks $P_1, \ldots, P_k$ ($k>0$) of total cost $L$ in $G$,
  there is a single $s$-$t$-walk $P$ in $G$ with cost $c(P)\leq L+d(k-1)\cdot\lp$ which contains all vertices of $P_1, \ldots, P_k$.
 \end{lemma}
 
 \begin{proof} 
  Let $\Iscr=(G,c,s,t)$ be an instance of ATSPP and $x^*$ be an optimum solution to \eqref{eq:subtour_lp_path}; so $\lp=c(x^*)$.
  We may assume that $G$ is the support graph of $x^*$.
  Consider the instance $\Iscr'=(G', c')$ of ATSP that arises from $\Iscr$ as follows.
  We add a new vertex $v$ to $G$ and two edges $(t,v)$ and $(v,s)$ with weights $c'(t,v) =d\cdot\lp$ and $c'(v,s)= 0$.
  Then there is a valid solution of the subtour LP for $\Iscr'$ with cost $(d+1)\cdot\lp$ 
  (extend $x^*$ by setting $x^*_{(t,v)}=x^*_{(v,s)}=1$). 
  Hence there is a solution to ATSP for $\Iscr'$ with cost at most $(d+1)\rhoatsp\cdot\lp$. 
  Let $R$ be such a solution. Then $R$ has to use $(t,v)$ and $(v,s)$ at least once, since it has to visit $v$. 
  By deleting all copies of $(t,v)$ and $(v,s)$ from $R$, we get $k>0$ $s$-$t$-walks in $G$ with total cost at most 
  $(d+1)\rhoatsp\cdot\lp-dk\cdot\lp$ such that every vertex of $G$ is visited by at least one of them. 
  Our assumption now guarantees the existence of a single $s$-$t$-walk $P$ with cost 
  $c(P)\leq (d+1)\rhoatsp\cdot\lp-dk\cdot\lp+d(k-1)\cdot\lp=((d+1)\rhoatsp - d)\cdot\lp$ in $G$, which contains every vertex of $G$. 
  This walk is a solution of ATSPP for $\Iscr$ and thus we have $\rhopath\le(d+1)\rhoatsp-d$ as proposed. 
 \end{proof}
 
The following procedure is similar to one step (``inducing on a tight set'') of the approximation algorithm for ATSP by 
\cite{SveTV18}. 

 \begin{lemma}\label{lemma:constructing_the_path}
   Let $(G,c,s,t)$ be an instance of ATSPP, where $G$ is the support graph of an optimum solution to \eqref{eq:subtour_lp_path}.
   Let $(a,y)$ be an optimum solution to \eqref{eq:dual_subtour_lp_path} with laminar support. 
 
  Let $k>0$ and $P_1, \ldots, P_k$ be $s$-$t$-walks  in $G$ with total cost $L$.  
  Then there is a single $s$-$t$-walk $P$ in $G$ which contains every vertex of $P_1, \ldots, P_k$ and 
  has cost at most $L+(k-1)(\lp+2(a_s-a_t))$.
 \end{lemma}

 \begin{proof}
  Let $V_1, \ldots, V_l$ be the vertex sets of the strongly connected components of $G$ in their topological order. 
  Let $P_i^j$ be the section of $P_i$ that visits vertices in $V_j$ (for $i=1, \ldots, k$ and $j=1, \ldots, l$).
  By Proposition \ref{structureoftightsets} applied to $U=V$, none of these sections of $P_i$ is empty.
  (Such a section might consist of a single vertex and no edges, but it has to contain at least one vertex.)

  We consider paths $R_i^j$ in $G$ for $j=1, \ldots, l$ that we will use to connect the walks 
  $P^j_1, \dots, P^j_k$ to a single walk visiting all vertices in $V_j$.
  See Figure \ref{figure:construction of P}.
  If $j$ is odd, let $R_i^j$ (for $i=1, \ldots k-1$) be a path from the last vertex of $P_i^j$ to the first vertex of $P_{i+1}^j$. 
  If $j$ is even, let $R_i^j$ (for $i=2, \ldots, k$) be a path from the last vertex of $P_i^j$ to the first vertex of $P_{i-1}^j$.
  (Such paths exists because $G[V_j]$ is strongly connected.)
  By Proposition \ref{prop:enterandleaveonlyonce} we can choose the paths $R_i^j$ such that they
  do not enter or leave any element of $\text{supp}(y)$ more than once.
    \begin{figure}[h]\label{figure:construction of P}
    \centering
    \begin{tikzpicture}[line cap=round,line join=round,x=1.3cm,y=1.3cm]
      \draw [color=red, line width=0.5pt, dashed] 
      (0.7,4.3)--(4,4.3) arc(90:0:0.3)--(4.3,3) arc(360:270:0.3)--(3.6,2.7) arc(90:180:0.3)--(3.3,0) 
      arc(360:270:0.3)--(0.7,-0.3) arc(270:180:0.3)--(0.4,4) arc(180:90:0.3);
      \draw [color=blue, line width=0.5pt, dashed] 
      (5,4.3)--(6,4.3) arc(90:0:0.3)--(6.3,3) arc(360:270:0.3)--(5.6,2.7) arc(90:180:0.3)--(5.3,1.6) arc(180:270:0.3)--(6,1.3)
      arc(90:0:0.3)--(6.3,0) arc(360:270:0.3)--(4,-0.3) arc(270:180:0.3)--(3.7,2) arc(180:90:0.3)--(4.4,2.3) arc(270:360:0.3)--(4.7,4)
      arc(180:90:0.3);
      \draw [color=darkgreen, line width=0.5pt, dashed] 
      (8,4.3) arc (90:270:0.3)--(9.4,3.7) arc(90:0:0.3)--(9.7,2.6) arc(360:270:0.3)--(9,2.3) arc(90:180:0.3)--(8.7,1.6) 
      arc(360:270:0.3)--(8,1.3) arc(90:180:0.3)--(7.7,0) arc(180:270:0.3)--(10.3,-0.3) arc(270:360:0.3)--(10.6,4) arc(0:90:0.3)--(8,4.3);
      \draw [line width=0.5pt, dashed] (6.4,2.3)--(6,2.3) arc(90:270:0.3)--(6.4,1.7) arc(90:45:0.3);
      \draw [line width=0.5pt, dashed] (6.4,2.3) arc(270:315:0.3);
      \draw [line width=0.5pt, dashed] 
      (7.6,3.3)--(9,3.3) arc(90:-90:0.3)--(8.6,2.7) arc(90:180:0.3)--(8.3,2) arc(360:270:0.3)--(7.6,1.7) arc(90:135:0.3);
      \draw [line width=0.5pt, dashed] (7.6,3.3) arc(270:225:0.3);

      \tikzstyle{rededge} = [color=red, line width=1.5pt, ->, >=latex, dotted]
      \tikzstyle{blueedge} = [color=blue, line width=1.5pt, ->, >=latex, dotted]
      \tikzstyle{greenedge} = [color=darkgreen, line width=1.5pt, ->, , >=latex, dotted]
      \draw [color=red, line width=1pt] (1,2.2)--(2,4)--(4,4) (1,2.07)--(2,3)--(4,3) (1,1.93)--(2,2)--(3,2) (1,1.8)--(2,0)--(3,0);
      \draw [rededge] (4,4)--(1,2.07);
      \draw [rededge] (4,3)--(1,1.93);
      \draw [rededge, dotted] (3,2)--(1,1.8);
      \draw [color=blue!50!red, line width=1pt] (3,0)--(4,0);
      \draw [color=blue, line width=1pt] (4,0)--(6,0) (4,2)--(5,2) (5,3)--(6,3) (5,4)--(6,4);
      \draw [blueedge, dotted] (6,0)--(4,2);
      \draw [blueedge] (5,2)--(5,3);
      \draw [blueedge] (6,3)--(5,4);
      \draw [color=blue!50!black, line width=1pt] (6,4)--(6.4,4);
      \draw [color=blue!50!black, line width=1pt, dotted] (6.4,4)--(6.7,4);
      \draw [color=darkgreen!50!black, line width=1pt, dotted] (7.2,4)--(7.6,4);
      \draw [color=darkgreen!50!black, line width=1pt] (7.6,4)--(8,4);
      \draw [color=darkgreen, line width=1pt] (8,4)--(9,4)--(10,2.2) (9,2)--(10,2) (8,0)--(9,0)--(10,1.8);
      \draw [greenedge] (10,2.2)--(9,2);
      \draw [greenedge, dotted] (10,2)--(8,0);

      \draw [color=black!70!white,line width=0.5pt] 
      (4,4)--(5,4) (4,3)--(5,3) (6,3)--(6.5,3) (7.5,3)--(8,3) (8,3)--(9,3) (9,3)--(10,2) (3,2)--(4,2) (5,2)--(6.5,2) (7.5,2)--(9,2) 
      (6,0)--(6.5,0) (7.5,0)--(8,0);
      \begin{large}
        \draw[color=black] (0.8,2) node {$s$};
        \draw[color=black] (10.2,2) node {$t$};
        \draw [color=red] (2.2,1) node {$V_1$};
        \draw [color=blue] (4.3,1) node {$V_2$};
        \draw [color=darkgreen] (8.4,1) node {$V_l$};

        \draw [fill=black] (2,0) circle (1.5pt);
        \draw [fill=black] (2,2) circle (1.5pt);
        \draw [fill=black] (2,3) circle (1.5pt);
        \draw [fill=black] (2,4) circle (1.5pt);
        \draw [fill=black] (3,4) circle (1.5pt);
        \draw [fill=black] (3,3) circle (1.5pt);
        \draw [fill=black] (3,2) circle (1.5pt);
        \draw [fill=black] (3,0) circle (1.5pt);
        \draw [fill=black] (4,4) circle (1.5pt);
        \draw [fill=black] (4,3) circle (1.5pt);
        \draw [fill=black] (4,2) circle (1.5pt);
        \draw [fill=black] (4,0) circle (1.5pt);
        \draw [fill=black] (5,4) circle (1.5pt);
        \draw [fill=black] (5,3) circle (1.5pt);
        \draw [fill=black] (5,2) circle (1.5pt);
        \draw [fill=black] (5,0) circle (1.5pt);
        \draw [fill=black] (6,4) circle (1.5pt);
        \draw [fill=black] (6,3) circle (1.5pt);
        \draw [fill=black] (6,2) circle (1.5pt);
        \draw [fill=black] (6,0) circle (1.5pt);
        \draw [fill=black] (8,4) circle (1.5pt);
        \draw [fill=black] (8,3) circle (1.5pt);
        \draw [fill=black] (8,2) circle (1.5pt);
        \draw [fill=black] (8,0) circle (1.5pt);
        \draw [fill=black] (9,4) circle (1.5pt);
        \draw [fill=black] (9,3) circle (1.5pt);
        \draw [fill=black] (9,2) circle (1.5pt);
        \draw [fill=black] (9,0) circle (1.5pt);
        \draw [color=black, line width=3pt] (10.05,2.2)--(10.05,1.8);
        \draw [color=black, line width=3pt] (0.95,2.2)--(0.95,1.8);
        
        \draw [fill=black!70!white] (6.8,4) circle (0.5pt);
        \draw [fill=black!70!white] (7,4) circle (0.5pt);
        \draw [fill=black!70!white] (7.2,4) circle (0.5pt);
        \draw [fill=black!70!white] (6.8,3) circle (0.5pt);
        \draw [fill=black!70!white] (7,3) circle (0.5pt);
        \draw [fill=black!70!white] (7.2,3) circle (0.5pt);
        \draw [fill=black!70!white] (6.8,2) circle (0.5pt);
        \draw [fill=black!70!white] (7,2) circle (0.5pt);
        \draw [fill=black!70!white] (7.2,2) circle (0.5pt);
        \draw [fill=black!70!white] (7,0) circle (0.5pt);
        \draw [fill=black!70!white] (7.2,0) circle (0.5pt);
        \draw [fill=black!70!white] (6.8,0) circle (0.5pt);
        %\draw [fill=black!70!white] (5.2,1) circle (0.5pt);
        %\draw [fill=black!70!white] (5.2,1.2) circle (0.5pt);
        %\draw [fill=black!70!white] (5.2,0.8) circle (0.5pt);
      \end{large}
    \end{tikzpicture}
    \caption{Construction of $P$. The $s$-$t$-walks $P_1, \dots, P_k$ are shown with solid lines.
    (Here, $P_1$ is the topmost walk and $P_k$ is shown in the bottom.)
    The vertex sets $V_1,\dots, V_l$ of the strongly connected components are indicated by the dashed lines.
    The red, blue, and green solid paths show the walks $P_i^j$, i.e. the sections of the walks $P_i$ within the 
    strongly connected components of $G$. The dotted arrows indicate the paths $R_i^j$.}
    \end{figure}
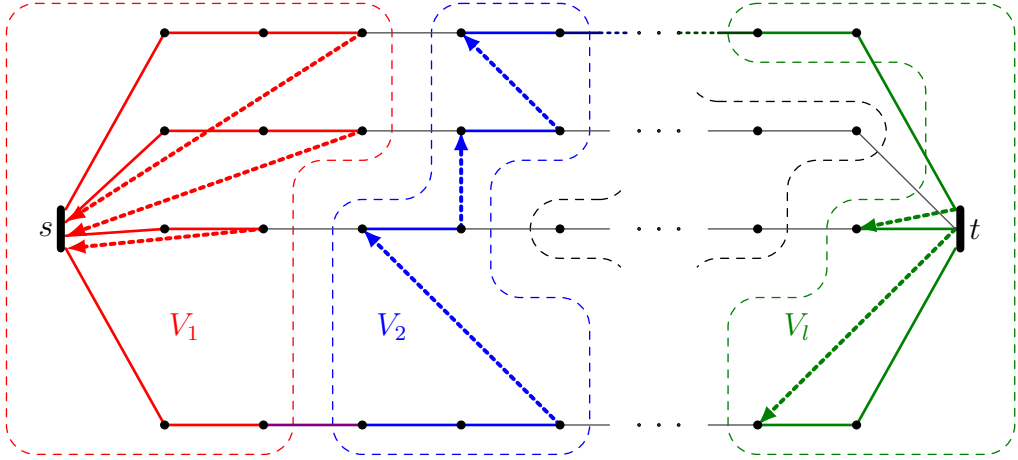
  
  We now costruct our $s$-$t$-walk $P$ that will visit every vertex of $P_1, \ldots, P_k$.
  We start by setting $P=s$ and then add for $j=1, \ldots, l$ all the vertices in $V_j$ to $P$ as follows.
  If $j$ is odd, we append $P_i^j$ and $R_i^j$ for $i=1$ to $i=k-1$ and at last $P_k^j$. 
  If $j$ is even, we append $P_i^j$ and $R_i^j$ for $i=k$ to $i=2$ and at last $P_1^j$.
  Note that when moving from one connected component $V_i$ to the next component $V_{i+1}$, we use an edge 
  from either $P_1$ (if $i$ is even) or $P_k$ (if $i$ is odd).
  Then $P$ is, indeed, an $s$-$t$ walk in $G$ and contains every vertex of $P_1, \ldots, P_k$. 
  We now bound the cost of the walk $P$.
  For every edge $e=(v,w)$ of $P$ we have by complementary slackness
  \begin{equation*}
   c(e) \ = \ a_w - a_v + \sum_{U:e\in\delta(U)}y_U.
  \end{equation*} 
  For an $s$-$t$-walk $R$ in $G$ we have 
  \begin{equation}
  \label{eq:costofwalk}
   c(R) \ = \sum_{(v,w)\in E(R)}\left(a_w -a_v +\sum_{U:(v,w)\in\delta(U)}y_U\right)\ = \ a_t - a_s + c^y(R),
  \end{equation}
  where the cost function $c^y$ is defined as $c^y(e):= \sum_{U:e\in\delta(U)}y_U$.
  Hence, to bound the cost of the $s$-$t$-walk $P$, we can bound $c^y(P)$ and then subtract $a_s$ and add $a_t$.

  $P$ is constructed from pieces of $P_1, \ldots, P_k$ and the paths $R_i^j$. 
  Each of the paths $R_i^j$ can only contain vertices of $V_j$. Two paths $R_i^j$ and $R_{i'}^{j'}$, 
  such that $j\neq j'$, can never both enter or both leave the same element of $\text{supp}(y)$:
  otherwise they would contain vertices of the same strongly connected component of $G$ by Proposition \ref{structureoftightsets}. 
  Thus every element of $\text{supp}(y)$ is entered at most $k-1$ times and left at most $k-1$ times on all the paths $R_i^j$ used in the construction of $P$, 
  and the total $c^y$ cost of these paths is at most $(k-1)\sum_{U}2y_U=(k-1)(\lp+a_s-a_t)$. 
  The $c^y$ cost of the edges of $P_1, \ldots, P_k$ is 
  \[ \sum_{i=1}^k c^y(P_i) \ = \ \sum_{i=1}^k\left(c(P_i)-a_t+a_s\right) \ = \ L+k\cdot a_s-k \cdot a_t.\]
  Consequently, we have 
  \begin{align*}
   c(P) &\ =\ a_t - a_s +  c^y(P)  \\
   &\ \leq \ a_t - a_s + L+k\cdot a_s-k \cdot a_t+(k-1) \bigl( \lp+a_s-a_t \bigr) \\
   & \ = \ L+(k-1) \bigl( \lp+2(a_s-a_t) \bigr)
  \end{align*}
  as claimed.
 \end{proof}
 
 \cite{SveTV18} reduced ATSP to so-called laminarly-weighted instances.
 In a laminarly-weighted instance we have $a=0$ (and $(a,y)$ has laminar support).
 For such instances Lemma \ref{lemma:bound_by_path} and Lemma \ref{lemma:constructing_the_path} would immediately 
 imply our main result (even with better constants).
 However, the reduction to laminarly-weighted instances for ATSP does not yield an analogous statement for the path version.
 Instead, we will prove that $a_s-a_t\le \lp$ for some optimum dual LP solution 
 (Section \ref{sect:bound_pot_difference}).
 
 Let us first consider a simpler special case.

 \begin{definition}
  An instance $(G,c,s,t)$ of ATSPP or an instance $(G,c)$ of ATSP is called \emph{node-weighted} 
  if there are nonnegative node weights $(c_v)_{v\in V}$ such that $c(v,w)=c_v+c_w$ for every edge $(v,w)$.
 \end{definition}
 
 Note that node-weighted instances are not necessarily symmetric because it might happen that an edge 
 $(v,w)$ exists, but $(w,v)$ does not exist.
 Since $x(\delta(s)) \ge 1$, $x(\delta(t))\ge 1$ and $x(\delta(v)) \ge 2$ for $v\notin \{s,t\}$ for every LP solution $x$,
 we have $\lp \ge c_s + c_t + \sum_{v\in V\setminus \{s,t\}}2c_v$.

 \begin{theorem}
  Let $\rhoatspNW$ be the integrality ratio for ATSP on node-weighted instances and 
  $\rhopathNW$ be the integrality ratio for ATSPP on node-weighted instances.
  Then 
  \begin{equation*}
   \rhopathNW \ \le \ 2 \rhoatspNW-1.
  \end{equation*}
 \end{theorem}

 \begin{proof}
  First we show how to modify the proof of Lemma \ref{lemma:bound_by_path} for node-weighted instances and $d=1$.
  For a node-weighted instance $\Iscr=(G,c)$, let $\Iscr'=(G',c')$ result from $\Iscr$ by adding a vertex 
  $v$ with weight $c_v=\frac{1}{2}(\lp-c_s-c_t)$ and two edges $(t,v)$ and $(v,s)$.
  Note that $\lp\ge c_s+c_t$ and hence $c_v\ge 0$. 
  Then continuing with the node-weighted instance $\Iscr'$ as in the proof of Lemma \ref{lemma:bound_by_path} yields the following:
  It suffices to show that for node-weighted instances of ATSPP,
  we can get an $s$-$t$-walk $P$ as in Lemma~\ref{lemma:constructing_the_path}, but with $c(P)\le L+(k-1)\lp$.
  
  We construct $P$ as in the proof of Lemma \ref{lemma:constructing_the_path}.
  Again we first bound the cost of the paths $R_i^j$. 
  For $1\le j\le l$ each vertex in $V_j$ can only be contained in the paths $R_i^j$
  (and not in a path $R_i^{j'}$ for $j' \ne j$).
  Hence every vertex can be used in at most $(k-1)$ paths $R_i^j$. 
  Furthermore, the vertices $s$ and $t$ are only used as the last or first vertices of paths $R_i^j$. 
  Hence the total cost of all paths $R_i^j$ can be bounded from above by 
  $$(k-1)\left(c_s+c_t +\sum_{v\in V\backslash \{s,t\}}2c_v\right) \ \le \ (k-1)\lp.$$ 
  This shows $c(P)\le L+(k-1)\lp$ and completes the proof.
 \end{proof}

\section{Bounding the difference of \boldmath $a_s$ and $a_t$}
\label{sect:bound_pot_difference}
  
  The goal of this section is to bound the difference of the dual variables $a_s$ and $a_t$ by $\lp$.
  Using Lemma \ref{lemma:bound_by_path} and Lemma \ref{lemma:constructing_the_path}, this will imply our main result $\rhopath\le 4\rhoatsp-3$.
  
  First, we give an equivalent characterization of the minimum value of $a_s-a_t$ in any optimum dual solution.
  This will not be needed to prove our main result, but might help to get some intuition.
 
  \begin{lemma}\label{lemma:characterize_a_s}
  Let $\Iscr = (G,c,s,t)$ be an instance of ATSPP and let $\Delta \ge 0$.
  Now consider the instance $\Iscr' = (G+e',c,s,t)$, where we add an edge $e'=(t,s)$ with $c(e'):= \Delta$.
  Then $\lp_{\Iscr} \ge \lp_{\Iscr'}$. 
  Moreover, $\lp_{\Iscr} = \lp_{\Iscr'}$ if and only if there exists an optimum solution $(a,y)$ 
  of \eqref{eq:dual_subtour_lp_path} for the instance $\Iscr$ with $a_s - a_t \le \Delta$.
 \end{lemma}
 \prove
 Every feasible solution $x$ of \eqref{eq:subtour_lp_path} for $\Iscr$ can be extended to a feasible 
 solution of \eqref{eq:subtour_lp_path} for $\Iscr'$ by setting $x_{e'}:=0$. This shows $\lp_{\Iscr} \ge \lp_{\Iscr'}$. 
 
 The dual LPs for the two instances are identical, except for the constraint corresponding to $e'$, which is
 \begin{equation}\label{eq:dual_costraint_e'}
  \Delta \ = \ c(e') \ \ge \ a_s - a_t + \sum_{\emptyset\ne U\subseteq V\setminus\{s,t\}, e'\in\delta(U)} y_U \ = \ a_s -a_t.
 \end{equation}

 Suppose $\lp_{\Iscr} = \lp_{\Iscr'}$.
 Let $(a,y)$ be an optimum dual solution for $\Iscr'$. 
 Then, \eqref{eq:dual_costraint_e'} is satisfied and $(a,y)$ is also feasible for the dual LP for the instance $\Iscr$.
 Moreover, since $\lp_{\Iscr} = \lp_{\Iscr'}$, the dual solution $(a,y)$ is also optimum for the instance $\Iscr$.
 
 For the reverse direction, let $(a,y)$ be an optimum solution
 to \eqref{eq:dual_subtour_lp_path} for the instance $\Iscr$ with $a_s - a_t \le \Delta$.
 Then $(a,y)$ satisfies \eqref{eq:dual_costraint_e'} and thus is also feasible for \eqref{eq:dual_subtour_lp_path} for $\Iscr'$.
 Hence, $\lp_{\Iscr'} \ge \lp_{\Iscr}$.
 \endproof
 
  We will need the following variant of Menger's Theorem.
  \begin{lemma}\label{lemma:menger_variant}
  Let $G$ be a directed graph and $s,t\in V(G)$ such that $t$ is reachable from $s$ in $G$. 
  Let $U\subseteq V(G) \setminus \{s,t\}$ such that 
  for every vertex $u\in U$, there exists an $s$-$t$-path in $G- u$.
  Then there exist two $s$-$t$-paths $P_1$ and $P_2$ in $G$ such that no vertex $u\in U$ is contained in
  both $P_1$ and $P_2$.
  \end{lemma}
  \prove 
  We construct a graph $G'$ that arises from $G$ as follows.
  We split every vertex $u\in U$ into two vertices $u^-$ and $u^+$ that are connected by an edge 
  $e_u := (u^-, u^+)$.
  Every edge $(v,u)$ is replaced by an edge $(v,u^-)$ and every edge $(u,v)$
  is replaced by an edge $(u^+,v)$. 
  In the graph $G'$ we now define integral edge capacities. 
  Every edge  $e_u$ for $u\in U$ has capacity one. 
  All other edges, i.e. all edges corresponding to edges of $G$, have capacity infinity.
  
  Since for every vertex $u\in U$, there exists an $s$-$t$-path in $G-u$,
  for every $u\in U$ there exists an $s$-$t$-path in $G' - e_u$.
  Thus, the minimum capacity of an $s$-$t$-cut in $G'$ is at least two.
  Hence, there exists an integral $s$-$t$-flow of value two in $G'$ with the defined edge capacities.
  This flow can be decomposed into two $s$-$t$-paths $P_1'$ and $P_2'$. 
  By the choice of the edge capacities no edge  $e_u$ for $u\in U$ occurs in both paths.
  Since this edge $e_u$ is the only outgoing edge of $u^-$ and the only incoming edge of $u^+$,
  an $s$-$t$-path using $u^-$ or $u^+$ must use $e_u$, and at most one of $P_1'$ and $P_2'$ can do so.
  
  Hence, contracting the edges $e_u$ (for $u\in U$) yields two $s$-$t$-paths $P_1$ and $P_2$ in $G$  
  such that no vertex $u\in U$ is contained in  both $P_1$ and $P_2$.
  \endproof
  
 We will now work with an optimum dual solution $(a,y)$ with $a_s-a_t$ minimum. 
 Note that this minimum is attained because for every feasible dual solution $(a,y)$ we have $a_s-a_t\ge -\lp$.
  
  \begin{lemma}\label{lemma:claim_1}
   Let $(G,c,s,t)$ be an instance of ATSPP, where $G$ is the support graph of an optimum solution to \eqref{eq:subtour_lp_path}.
 Let $(a,y)$ be an optimum solution of \eqref{eq:dual_subtour_lp_path} such that $a_s - a_t$ is minimum. 
 Let $\bar U\subseteq V\setminus\{s,t\}$ such that every $s$-$t$-path in $G$ enters (and leaves) $\bar U$ at least once.   
   Then $y_{\bar U}=0$.
  \end{lemma}
  \prove
  Suppose $y_{\bar U} >0$ and let $\epsilon :=y_{\bar U}$.
 Let $R$ be the set of vertices reachable from $s$ in $G -\bar U$.
 We define a dual solution $(\bar a, \bar y)$ as follows:
 \begin{align*}
  \bar y (U) :=&\ \begin{cases}
                   y_U - \epsilon &\text{ if }U = \bar U \\
                   y_U &\text{ else}
                  \end{cases} \\
  \bar a_v :=&\ \begin{cases}
                   a_v - 2\epsilon &\text{ if }v\in R \\
                   a_v - \epsilon &\text{ if }v \in \bar U \\
                   a_v  &\text{ else}.
                  \end{cases} 
 \end{align*}
 See Figure \ref{fig:modifying_dual_solution}.
 \begin{figure}
 \begin{center}
  \begin{tikzpicture}
    \tikzstyle{vertex}=[circle,fill,minimum size=3,inner sep=0pt] 
    \tikzstyle{edge}=[->, >=latex, thick]
    \fill[darkgreen, opacity=0.3] (2,0) ellipse ({2.2} and {1.2});
    \fill[blue, opacity=0.2] (5.5,0) ellipse ({0.8} and {1.2});
    \draw[red, very thick] (5.5,0) ellipse ({0.8} and {1.2});
    \fill[gray, opacity=0.3] (9,0) ellipse ({2.2} and {1.2});
    \node[darkgreen] at (2,1.5) {$R$};
    \node[blue] at (5.5,1.5) {$\bar U$};
    \node[gray] at (9,1.5) {$V\setminus(R \cup \bar U)$};
    \node[red] at (6.4, -0.9) {$-\epsilon$};
    \node[darkgreen] at (2,-1.5) {$-2\epsilon$};
    \node[blue] at (5.5,-1.5) {$-\epsilon$};
    \node[vertex](s) at (0.5,0){};
    \node[left] at (s) {$s$};
    \node[vertex](t) at (10.5,0){};
    \node[right] at (t) {$t$};
  \end{tikzpicture}
  \end{center}
  \caption{Modifying the dual solution in the proof of Lemma \ref{lemma:claim_1}. 
  The green and blue numbers in the bottom indicate the 
  change of the dual node variables. In red the change of the variable $y_{\bar U}$ is indicated. 
  There is no edge from $R$ to $V\setminus (R \cup \bar U)$.
  \label{fig:modifying_dual_solution}
  }
 \end{figure}
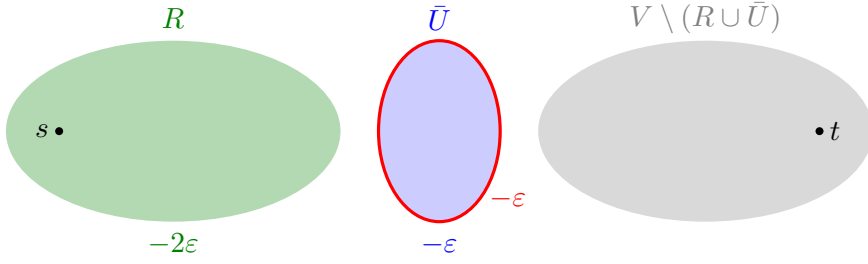
 We claim that $(\bar a, \bar y)$ is an optimum (and feasible) solution to \eqref{eq:dual_subtour_lp_path}. 
 Note that $t\in V \setminus \left( R \cup \bar U\right)$ and thus $\bar a_t = a_t$.
 Since $s\in R$, we have $\bar a_s - \bar a_t < a_s - a_t$.
 Thus, if $(\bar a, \bar y)$ is indeed optimum (and feasible), 
 we obtain a contradiction to our choice of the dual solution $(a,y)$.

 First, we observe that  $(\bar a, \bar y)$ and $(a,y)$ have the same objective value since
 \[ \bar a_t- \bar a_s +  \sum_{\emptyset \ne U \subseteq V\setminus \{s,t\}} 2\bar y_U \ =\
    a_t - \left(a_s - 2 \epsilon\right)  +\sum_{\emptyset \ne U \subseteq V\setminus \{s,t\}} 2 y_U - 2 \epsilon .
 \] 
 By our choice of $\epsilon$, the vector $\bar y$ will be non-negative. Now consider an edge $e=(v,w)\in E(G)$.
 We need to show that 
 \begin{equation}\label{eq:dual_constraint}
  \bar a_w - \bar a_v + \sum_{U : e\in \delta(U)} \bar y_U\ \le\ c(e).
 \end{equation}
 To prove this we will show that 
 \begin{equation}\label{eq:sufficient_for_dual_feasible}
   \bar a_w - a_w - \bar a_v + a_v + \sum_{U : e\in \delta(U)} \left(\bar y_U - y_U\right) \le 0.
 \end{equation}
 Since $(a,y)$ is a feasible dual solution, this will imply \eqref{eq:dual_constraint}.
 We have 
  \begin{align*}
  \bar a_w - a_w :=&\ \begin{cases}
                   - 2\epsilon &\text{ if }w\in R \\
                    - \epsilon &\text{ if }w \in \bar U \\
                   0  &\text{ else,}
                  \end{cases} \\
    - \bar a_v + \bar a_v :=&\ \begin{cases}
                   2\epsilon &\text{ if }v\in R \\
                     \epsilon &\text{ if }v \in \bar U \\
                   0  &\text{ else,}
                  \end{cases}                \\
     \sum_{U : e\in \delta(U)} \left(\bar y_U - y_U\right) :=&\ \begin{cases}
                    - \epsilon &\text{ if }(v,w)\in \delta(\bar U) \\
                   0 &\text{ else}.
                  \end{cases}             
 \end{align*}
 Since $\bar a_w - a_w\le 0$ and $\bar  \sum_{U : e\in \delta(U)} \left(\bar y_U - y_U\right) \le 0$,
 it suffices to consider the cases $v\in R$ and $v\in \bar U$.
 If $v\in R$, we have by definition of $R$, either $w\in R$ or $w\in \bar U$.
 In both cases \eqref{eq:dual_constraint} holds, because if $w\in \bar U$, we have $(v,w)\in \delta(\bar U)$.
 Now let $v\in \bar U$. Then if $(v,w)\in \delta(\bar U)$, we have $ \sum_{U : e\in \delta(U)} \left(\bar y_U - y_U\right) = -\epsilon$,
 implying \eqref{eq:sufficient_for_dual_feasible}.
 Otherwise, $w\in \bar U$ and $\bar a_w - a_w- \bar a_v + a_v = 0$.
 
 This shows that $(\bar a, \bar y)$ is an optimum dual solution and $\bar a_s - \bar a_t < a_s - a_t$, a contradiction. 
 Hence, $y_{\bar U}=0$.
  \endproof

   We will now continue to work with a dual solution $(a,y)$ that minimizes $a_s -a_t$.
   By Proposition \ref{duallaminar}, we can assume in addition that $(a,y)$ has laminar support.
  
 \begin{lemma}\label{lemma:claim_2}
  Let $(G,c,s,t)$ be an instance of ATSPP, where $G$ is the support graph of an optimum solution to \eqref{eq:subtour_lp_path}.
 Let $(a,y)$ be an optimum solution to \eqref{eq:dual_subtour_lp_path} that has laminar support and minimum $a_s - a_t$. 
   
 Then $G$ contains two $s$-$t$-paths $P_1$ and $P_2$ 
 such that for every set $U\in \text{supp}(y)$ we have $|E(P_1) \cap \delta(U)| + |E(P_2) \cap \delta(U)| \le 2$.
   \end{lemma}
  \prove
 By Lemma \ref{lemma:claim_1}, for every set $U\in \text{supp}(y)$ there is an $s$-$t$-path in $G$ that visits no vertex in $U$.
 We contract all maximal sets $U\in \text{supp}(y)$. 
 Using Lemma \ref{lemma:menger_variant},
 we can find two $s$-$t$-paths in $G$ such that each vertex arising from the contraction of a set 
 $U\in \text{supp}(y)$ is visited by at most one of the two paths.
 
 Now we revert the contraction of the sets $U\in \text{supp}(y)$. 
 We complete the edge sets of the two $s$-$t$-paths
 we found before (which are not necesarily connected anymore after undoing the contraction), to
 paths $P_1$ and $P_2$ with the desired properties.
 To see that this is possible, let $v$ be the end vertex of an edge entering a contracted set $U\in \text{supp}(y)$ 
 and let $w$ be the start vertex of an edge leaving $U$.
 Then by Proposition \ref{structureoftightsets}, the vertex $w$ is reachable from $v$ in $G[U]$ and by 
 Proposition \ref{prop:enterandleaveonlyonce}, we can choose a $v$-$w$-path in $G[U]$ that enters and leaves every set 
 $U'\in \text{supp}(y)$ with $U' \subsetneq U$ at most once.
\endproof

  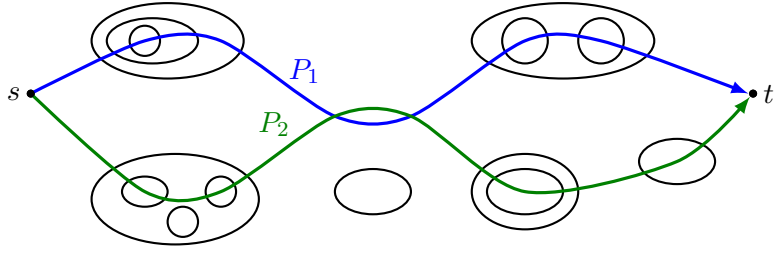
\begin{figure}
  \begin{center}
   \begin{tikzpicture}
        \tikzstyle{vertex}=[circle,fill,minimum size=3,inner sep=0pt] 
    \tikzstyle{edge}=[->, >=latex, thick]
    \node[vertex](s) at (0.5,0.3){};
    \node[left] at (s) {$s$};
    \node[vertex](t) at (10,0.3){};
    \node[right] at (t) {$t$};
    
    \draw[black, thick] (2,1) ellipse ({0.2} and {0.2});
    \draw[black, thick] (2.1,1) ellipse ({0.6} and {0.3});
    \draw[black, thick] (2.3,1) ellipse ({1} and {0.5});
       
     \draw[black, thick] (7,1) ellipse ({0.3} and {0.3});
     \draw[black, thick] (8,1) ellipse ({0.3} and {0.3});
     \draw[black, thick] (7.5,1) ellipse ({1.2} and {0.5});
     
      \draw[black, thick] (2,-1) ellipse ({0.3} and {0.2});
     \draw[black, thick] (3,-1) ellipse ({0.2} and {0.2});
     \draw[black, thick] (2.5,-1.4) ellipse ({0.2} and {0.2});
    \draw[black, thick] (2.4,-1.1) ellipse ({1.1} and {0.6});
    
      \draw[black, thick] (5,-1) ellipse ({0.5} and {0.3});
      \draw[black, thick] (7,-1) ellipse ({0.5} and {0.3});
      \draw[black, thick] (7,-1) ellipse ({0.7} and {0.5});
      \draw[black, thick] (9,-0.6) ellipse ({0.5} and {0.3});
      
     % \draw [->, >=latex, red, thick] plot [smooth] coordinates {(t)(8,1.8) (2,1.8) (0.51, 0.36)}; 
     % \node[red] at (5,1.7) {$e'$};
      \node[blue] at (4.1,0.6) {$P_1$};
      \node[darkgreen] at (3.7,-0.1) {$P_2$};
     
     \draw [->, >=latex, blue, very thick] plot [smooth] coordinates {(s) (2,1) (3,1) (4.5,0) (5.5,0) (7,1) (8,1) (9.94,0.3)};  
     \draw [->, >=latex, darkgreen,very thick] plot [smooth] coordinates {(s) (2,-1) (3,-1) (4.5,0) (5.5,0) (7,-1) (9,-0.6) (9.96,0.26)};  
     \node[vertex](s) at (0.5,0.3){};
     \node[vertex](t) at (10,0.3){};
   \end{tikzpicture}
   \end{center}
  \caption{The paths $P_1$ and $P_2$ as in Lemma \ref{lemma:claim_2}.
  In black the vertex sets $U\in \text{supp}(y)$ are shown. 
  The paths $P_1$ and $P_2$ are not necesarily disjoint but they never both cross the same set $U$ with $y_U>0$. \label{fig:paths_p1_and_p_2}}
  \end{figure}

We finally show our main lemma.

\begin{lemma}\label{lemma:bound_pot_difference}
 Let $\Iscr=(G,c,s,t)$ be an instance of ATSPP, where $G$ is the support graph of an optimum solution to \eqref{eq:subtour_lp_path}.
 Then there is an optimum solution $(a,y)$ of \eqref{eq:dual_subtour_lp_path} 
 with laminar support and $a_s -a_t \le \lp$.
\end{lemma}

\prove
Let $(a,y)$ be an optimum solution to \eqref{eq:dual_subtour_lp_path} that has laminar support and minimum $a_s - a_t$. 
Note that such an optimum dual solution exists by Proposition \ref{duallaminar}.
We again define the $c^y$ cost of an edge $e$ to be $c^y(e) = \sum_{U: e\in \delta(U)} y_U$.
By Lemma \ref{lemma:claim_2}, $G$ contains two $s$-$t$-paths $P_1$ and $P_2$ 
such that $c^y(P_1) + c^y(P_2) \le \sum_{\emptyset \ne U \subseteq V\setminus \{s,t\}} 2\cdot y_U$.
Then, using \eqref{eq:costofwalk},
\begin{align*}
 0 \ &\le\ c(P_1) + c(P_2) \\
 &=\ c^y(P_1) - (a_s - a_t) + c^y(P_2) - (a_s - a_t) \\
 &\le\  \sum_{\emptyset \ne U \subseteq V\setminus \{s,t\}} 2\cdot y_U - 2 (a_s - a_t),
\end{align*}
implying
\[ a_s -a_t \ \le \ \sum_{\emptyset \ne U \subseteq V\setminus \{s,t\}} 2\cdot y_U -  (a_s - a_t) \ = \ \lp. \]
\endproof

  	\begin{figure}
	\centering
	\begin{tikzpicture}[xscale=2]
	  \tikzstyle{vertex}=[circle,fill,minimum size=3,inner sep=2pt, outer sep=2pt] 
          \tikzstyle{edge}=[->, >=latex, thick]
          \node[vertex] (s) at (0,0) {};
          \node[vertex] (t) at (2,0) {};
          \node[vertex] (v) at (1,-1) {};
          \node[vertex] (w) at (1,1) {};
          \node () at (0.4,0.7) {\small$0$};
          \node () at (1.6,0.7) {\small$0$};
          \node () at (1.6,-0.7) {\small$0$};
          \node () at (0.4,-0.7) {\small$0$};
          \node () at (0.9,0) {\small$1$};
          \node[left=1pt] () at (s) {\Large$s$};
          \node[right=1pt] () at (t) {\Large$t$};
          \draw[edge] (v) to (w);
          \draw[edge] (s) to (v);
          \draw[edge] (s) to (w);
          \draw[edge] (v) to (t);
          \draw[edge] (w) to (t);
	\end{tikzpicture}
	\caption{
		Example with no optimum dual solutions with $a_s-a_t<\lp$: 
		The numbers next to the arcs denote their cost. For this instance we have $\lp=1$. 
		However adding an edge $(t,s)$ with cost $\gamma<1$ would result in an instance with $\lp=\gamma$. 
		By Lemma \ref{lemma:characterize_a_s} there cannot be an optimum dual solution where $a_s-a_t<1=\lp$.
		\label{fig:pot_diff_tight}}
	\end{figure}
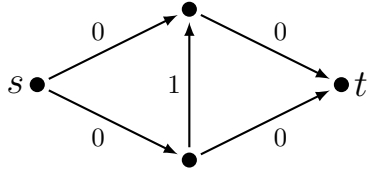
  
  We remark (although we will not need it) that Lemma \ref{lemma:bound_pot_difference} also
  holds for general instances. To adapt the proof, work with the subgraph $G'$ of $G$ that
  contains all edges of $G$ for which the dual constraint is tight. Now $G'$ plays the role of $G$ in the proof,
  and by choosing $\epsilon$ small enough in the proof of Lemma \ref{lemma:claim_1}
  we maintain dual feasibility also for the edges that are not in $G'$.
  
  By Lemma \ref{lemma:characterize_a_s}, this also shows that
  adding an edge $(t,s)$ of cost equal to the LP value does not change the value of an optimum LP solution. 
  
  The instance in Figure \ref{fig:pot_diff_tight} shows that the bound $a_s-a_t \le \lp$ is tight.
  Note that the bound is also tight for the instance in Figure \ref{fig:example} in which $x^*_e>0$ for all edges $e$, 
  and in which the integrality ratio is arbitrarily close to the best known lower bound of $2$.
 
 We will now prove our main result.
 
  \begin{theorem}\label{thm:main_result_integrality_gap}
  Let $\rhoatsp$ be the integrality ratio of \eqref{eq:subtour_lp_atsp}.
  Then the integrality ratio $\rhopath$ of \eqref{eq:subtour_lp_path} is at most
  $4\rhoatsp - 3$.
 \end{theorem}
 \prove 
 Let $(G,c,s,t)$ be an instance of ATSPP, where $G$ is the support graph of an optimum solution to \eqref{eq:subtour_lp_path}.
 By Lemma \ref{lemma:bound_pot_difference}, 
 there is an optimum dual solution $(a,y)$ with laminar support and $a_s-a_t \le \lp$.
 Using Lemma \ref{lemma:constructing_the_path}, this implies that the condition of
 Lemma \ref{lemma:bound_by_path} is fulfilled for $d=3$. 
 This shows $\rhopath \le 4\rhoatsp - 3$.
 \endproof

\section{Node-weighted and unweighted instances}\label{sec:nodeweighted=unweighted} 

Here we observe that, for ATSP, 
node-weighted instances are not much more general than unweighted instances.
We call an LP solution $x$ minimal if there is no feasible solution $x'\ne x$ with $x'\le x$ componentwise.

\begin{lemma}\label{lemma:minimalsolutions}
 For every minimal solution $x$ of \eqref{eq:subtour_lp_atsp}, we have $x(E(G))\le n^2$, where $n=|V(G)|$.
\end{lemma}
\prove
Choose an arbitrary root $r\in V$ and let 
$P=\{y \in \mathbb{R}^{E(G)}_{\ge 0}: y(\delta^-(U))\ge 1 \text{ for } \emptyset\not=U\subseteq V\setminus\{r\}\}$.
A vector is feasible for \eqref{eq:subtour_lp_atsp} if and only if it is a circulation that belongs to $P$.
Let $y\le x$ be a minimal vector in $P$.
The minimal vectors in $P$ are the convex combinations of incidence vectors of spanning arborescences rooted at $r$ (\cite{Edm67});
hence $y(E(G))=n-1$.
There are cycles $C_j$ and edge sets $S_j \subseteq C_j$ ($j=1,\ldots,l$) such that
$x=\sum_{j=1}^l \lambda_j \chi^{C_j}$ and $y= \sum_{j=1}^l \lambda_j \chi^{S_j}$ for some positive coefficients $\lambda_j$.
Note that none of the sets $S_j$ can be empty because otherwise $x'= x -  \lambda_j \chi^{C_j}$ would be a circulation that belongs to $P$,
contradicting the minimality of $x$.
We conclude
$ x(E(G)) = \sum_{j=1}^l \lambda_j |C_j| \le \sum_{j=1}^l  \lambda_j \cdot n |S_j| = n \cdot y(E(G)) = n(n-1)$.
\endproof

\begin{lemma}\label{lemma:unweighted}
Let $\epsilon>0$.
Let $(G,c)$ be a node-weighted instance of ATSP with $n$ vertices.
Then we can find in polynomial time a constant $M>0$ and an unweighted digraph $G'$ with $O(\frac{n^2}{\epsilon})$ vertices such that
\begin{enumerate}[(i)]
 \item $\lp_{(G,c)} \le M \cdot \lp_{G'} \le (1+\epsilon) \lp_{(G,c)}$, 
 \item $\opt_{(G,c)} \le M \cdot \opt_{G'} \le (1+\epsilon) \opt_{(G,c)}$, and 
 \item for every tour $F'$ in the unweighted digraph $G'$ there is a corresponding tour $F$ in $G$ such that $c(F) \le M |F'|$
       and $F$ can be obtained from $F'$ in polynomial time.
\end{enumerate}
\end{lemma}

\prove
Let $c_v\ge 0$ ($v\in V(G)$) be the node weights, i.e., $c(v,w)=c_v+c_w$ for all $(v,w)\in E$.
Let $c(V(G))=\sum_{v\in V(G)} c_v$ denote the sum of all node weights. 
If $c(V(G))=0$, the instance is trivial, we can choose $G'$ to consist of a single vertex.

Otherwise let $n=|V(G)|$, $M:=\frac{2\epsilon \cdot c(V(G))}{n^2}$
and $\bar c_v := \lfloor\frac{2c_v}{M}\rfloor$ for all $v\in V(G)$.
Replace every vertex $v$ of $G$ with $\bar c_v > 0$ by two vertices $v^-$ and $v^+$, such that $v^-$ inherits the entering edges and $v^+$ inherits the outgoing edges,
and add a path $P_v$ of $\bar c_v$ edges from $v^-$ to $v^+$. This defines $G'$. 
Note that $|V(G')| = n + \sum_{v\in V(G)}\bar c_v \le n+\frac{n^2}{\epsilon}$.

Every solution $x$ to \eqref{eq:subtour_lp_atsp} for $(G,c)$ corresponds to a solution $x'$ to \eqref{eq:subtour_lp_atsp} for $G'$, simply
by setting $x'_e:=x(\delta^+(v))$ for all edges $e$ of $P_v$. 
Then
\begin{align*}
 x'(E(G')) \ = \ & \sum_{v\in V(G)}  \left(1+\bar c_v\right) \, x(\delta^+(v)) \\
\ = \ & \sum_{v\in V(G)}  \left(1+\left\lfloor\frac{2c_v}{M}\right\rfloor\right) \, x(\delta^+(v))  \\
\ = \ & \delta \cdot x(E(G)) + \sum_{v\in V(G)}  \frac{2c_v}{M}  \, x(\delta^+(v)) \\
\ = \ & \delta  \cdot x(E(G))+ \frac{1}{M}  c(x)
\end{align*}
for some $\delta \in [0,1]$. Hence
$$ c(x) \ \le \ M x'(E(G)),$$
and for minimal solutions we have $x(E(G)) \le n^2$ by Lemma \ref{lemma:minimalsolutions}, which implies 
 $\delta  \cdot x(E(G)) \le n^2 = \epsilon \frac{2c(V(G))}{M}\le \epsilon \frac{c(x)}{M}$ and thus
 $$M x'(E(G))\ \le \ (1+\epsilon) c(x).$$
Because tours are integral LP solutions, and optimum LP solutions and optimum tours can be assumed to be minimal,
this completes the  proof of (i) and (ii).
To prove (iii), observe that contracting the paths $P_v$ in a tour $F'$ yields a tour $F$ as claimed.
\endproof

This immediately implies:

\begin{theorem}
The integrality ratio of \eqref{eq:subtour_lp_atsp} is the same for unweighted and for node-weighted instances.
For any constants $\alpha\ge 1$ and $\epsilon>0$, there is a polynomial-time $(\alpha+\epsilon)$-approximation algorithm for node-weighted instances
if there is a polynomial-time $\alpha$-approximation algorithm for unweighted instances.
\end{theorem}
\prove 
The equality of the integrality ratio for unweighted and for node-weighted instances follows from 
Lemma \ref{lemma:unweighted} (i) and (ii).
Now suppose we have a polynomial-time $\alpha$-approximation algorithm for unweighted instances.
Then for a node-weighted instance $(G,c)$ we apply Lemma \ref{lemma:unweighted} with $\epsilon'=\sfrac{\epsilon}{\alpha}$ 
and apply our $\alpha$-approximation algorithm to the resulting digraph $G'$. Let $F'$ be the resulting tour in $G'$.
By (iii) of Lemma \ref{lemma:unweighted}, this tour corresponds to a tour $F$ in $G$ such that 
\[ c(F) \ \le \ M|F'| \ \le \ \alpha\cdot M \opt_{G'} \ \le \ (1+\epsilon')\alpha\cdot \opt_{(G,c)} \ = \ (\alpha+\epsilon)\cdot \opt_{(G,c)}.\]
\endproof

\begin{figure}
\begin{center}
 \begin{tikzpicture}[xscale=0.9]
  \def\arccolor{black}
  \def\pathcolor{blue}
  \def\gadgetcolor{darkgreen}
  \def\height{3}
  \tikzstyle{vertex}=[circle,fill,minimum size=3,inner sep=0,outer sep=1]
  \tikzstyle{edge}=[->, >=latex, \arccolor]
  \tikzset{snake arrow/.style=
{->,
decorate,
decoration={snake,amplitude=.4mm,segment length=2mm,post length=1mm}}
}
% G_0
\def\posG{4.5}
\node (G0) at (1.5, \posG) {$G_0$};
 \foreach [evaluate=\i as \xcoord using 5+2*\i] \i in {0, ..., 4} {
    \node[vertex] (g\i) at (\xcoord,\posG) {};
 }
 \foreach [evaluate=\i as \iplusone using int(1+\i)] \i in {0, ..., 3} {
    \draw[edge, bend left] (g\i) to (g\iplusone);
    \draw[edge, bend left] (g\iplusone) to (g\i);
 }
 \node[left=1pt] () at (g0) {$v_0 =v_0'$};
 \node[right=1pt] () at (g4) {$w_0 =w_0'$};
%  G_i
\node (Gi) at (1.5, 1.5) {$G_i$};
% gadgets
 \foreach [evaluate=\i as \leftcorner using 1+3*\i] \i in {1, ..., 4} {
    \begin{scope}[shift={(\leftcorner,0)}, \gadgetcolor]
      \node[vertex] (w\i) at (0,0) {};
      \node[vertex] (w'\i) at (1,0) {};
      \node[vertex] (v\i) at (0,\height) {};
      \node[vertex] (v'\i) at (1,\height) {};
      \fill[opacity=0.3] (0,0) rectangle (1,\height);
      \node[black] (name) at (0.5,1.5){\small $G_{i-1}$};
    \end{scope}
 }
 \begin{scope}[\gadgetcolor]
 \foreach [evaluate=\i as \leftcorner using 1+3*\i, evaluate=\i as \rightcorner using 2+3*\i] \i in {1,3} {
    \node[above=2pt] () at (\leftcorner,\height) {\scriptsize$v'_{i-1}$};
    \node[above=2pt] () at (\rightcorner,\height) {\scriptsize$v_{i-1}$}; 
    \node[below=3pt] () at (\leftcorner,0) {\scriptsize$w_{i-1}$};
    \node[below=1pt] () at (\rightcorner,0) {\scriptsize$w'_{i-1}$}; 
 }
 \foreach [evaluate=\i as \leftcorner using 1+3*\i,  evaluate=\i as \rightcorner using 2+3*\i] \i in {2,4} {
    \node[above=2pt] () at (\rightcorner,\height) {\scriptsize$v'_{i-1}$};
    \node[above=2pt] () at (\leftcorner,\height) {\scriptsize$v_{i-1}$}; 
    \node[below=3pt] () at (\rightcorner,0) {\scriptsize$w_{i-1}$};
    \node[below=1pt] () at (\leftcorner,0) {\scriptsize$w'_{i-1}$}; 
 }
\end{scope}
%vertices and wiggly paths
\foreach  [evaluate=\i as \xcoord using 3*\i] \i in {1,2,3,4,5}{
   \node[vertex, \pathcolor] (o\i) at (\xcoord, \height) {};
   \node[vertex, \pathcolor] (u\i) at (\xcoord, 0) {};
      \pgfmathparse{int(mod(\i, 2))}
      \ifthenelse{\equal{\pgfmathresult}{0}}{
         \draw[edge, \pathcolor, snake arrow] (u\i) to (o\i);
      } {
         \draw[edge, \pathcolor, snake arrow] (o\i) to (u\i);
      }
 }
 \node[left=1pt] () at (o1) {$v_i$};
 \node[left=1pt] () at (u1) {$v'_i$};
  \node[right=1pt] () at (o5) {$w_i$};
 \node[right=1pt] () at (u5) {$w'_i$};
 
 % graph edges
 \foreach  [evaluate=\i as \xcoord using 3*\i, evaluate=\i as \iplusone using \i +1, evaluate] \i in {1,2,3,4}{
    \pgfmathparse{int(mod(\i, 2))}
      \ifthenelse{\equal{\pgfmathresult}{0}}{
         \draw[edge] (o\i) to (v\i);
         \draw[edge] (w\i) to (u\i);
      } {
         \draw[edge] (v\i) to (o\i);
         \draw[edge]  (u\i) to (w\i);
      }
 }
 \foreach  [evaluate=\i as \xcoord using 3*\i, evaluate=\i as \iminusone using \i -1] \i in {2,3,4,5}{
      \pgfmathparse{int(mod(\i, 2))}
      \ifthenelse{\equal{\pgfmathresult}{0}}{
         \draw[edge] (o\i) to (v'\iminusone);
         \draw[edge] (w'\iminusone) to (u\i);
      } {
         \draw[edge] (v'\iminusone) to (o\i);
         \draw[edge] (u\i) to (w'\iminusone);
      }
 }
 \end{tikzpicture}
\end{center}
\caption{
Constructing a family of digraphs with integrality ratio arbitrary close to $2$ for ATSP with unit weights.
For a fixed even number $l\ge 4$ we define graphs $G_0, G_1,\dots $. 
The graph $G_0$ consists of a bidirected path of length $l$.
Then we construct $G_i$ from $G_{i-1}$ as in the picture. 
The picture shows the construction for $l=4$; in general, there are $l$ copies of the graph $G_{i-1}$ (shown in green). 
The blue wiggly paths indicate paths of length $d_i$, where $d_0 =0$ and $d_i = l^{i} -d_{i-1} -2$.
Let $G_i'$ be the graph arising from $G_i$ by identifying the blue $v_i$-$v'_i$-path with the blue 
$w_i$-$w'_i$-path. Then for $i\rightarrow \infty$, the integrality ratio of $G_i'$ converges to $2-\frac{2}{l}$
(\cite{BoyEM15}). \label{fig:construction_digraph_example}
}
\end{figure}
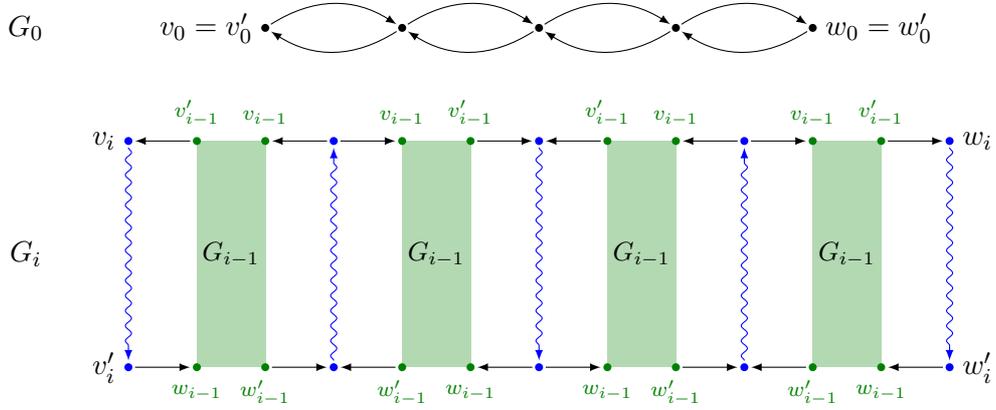

\begin{figure}
\begin{center}
 \begin{tikzpicture}[>=latex,  every node/.style={fill, circle, inner sep =1pt}, scale=0.7]

\def\bidirectedspokelength{7}
\def\onewayspokelength{5}
\def\nodedistancescale{3}
\pgfmathparse{(\bidirectedspokelength - 1) * (\onewayspokelength - 1) / \nodedistancescale}
\edef\spokelengthproduct{\pgfmathresult}
      
\def\circularcolor{black}
\def\onewayspokecolor{blue}
\def\bidirectedspokecolor{darkgreen}
      
\foreach [evaluate=\i as \degree using 30*\i] \i in {0, 1, ..., 12} {
      \pgfmathparse{int(mod(\i, 2))}
      \ifthenelse{\equal{\pgfmathresult}{0}}{
	      \foreach [evaluate=\j as \radius using \onewayspokelength*(\j -1)/\spokelengthproduct + 1] \j in {1, 2, ..., \bidirectedspokelength} {
		      \node[\bidirectedspokecolor] (v\i\j) at (\degree:\radius) {};
	      }
      } {
	      \foreach [evaluate=\j as \radius using \bidirectedspokelength*(\j -1)/\spokelengthproduct + 1] \j in {1, 2, ..., \onewayspokelength} {
		      \node[\onewayspokecolor] (v\i\j) at (\degree:\radius) 	{};
	      }
      }
}
      
\foreach \i in {0, 1, ..., 12} {
      \pgfmathparse{int(mod(\i, 2))}
      \ifthenelse{\equal{\pgfmathresult}{0}}{
	      \foreach [evaluate=\j as \jminusone using int(\j -1)] \j in {2, 3, ..., \bidirectedspokelength} {
		      \draw[->, color=\bidirectedspokecolor] (v\i\jminusone) to[bend left=30] (v\i\j);
		      \draw[->, color=\bidirectedspokecolor] (v\i\j) to[bend left] (v\i\jminusone);
	      }
      }{
	      \foreach [evaluate=\j as \jminusone using int(\j -1)] \j in {2, 3, ..., \onewayspokelength} {
		      \pgfmathparse{int(mod(\i,4))}
		      \ifthenelse{\equal{\pgfmathresult}{3}} {
			      \draw[->, color=\onewayspokecolor] (v\i\jminusone) to (v\i\j);
		      } {
			      \draw[->, color=\onewayspokecolor] (v\i\j) to (v\i\jminusone);
		      }
	      }
      }
      
      \pgfmathparse{int(mod(\i +1,12))}
      \edef\iplusone{\pgfmathresult}
     
      \pgfmathparse{int(mod(\i,4))}
      \ifthenelse{\equal{\pgfmathresult}{0}} {
	      \draw[->, color=\circularcolor] (v\iplusone1) to (v\i1);
	      \draw[->, color=\circularcolor] (v\i\bidirectedspokelength) to (v\iplusone\onewayspokelength);
      } {}
      
      \ifthenelse{\equal{\pgfmathresult}{1}} {
	      \draw[->, color=\circularcolor] (v\i1) to (v\iplusone1);
	      \draw[->, color=\circularcolor] (v\iplusone\bidirectedspokelength) to (v\i\onewayspokelength);
      } {}
     
      \ifthenelse{\equal{\pgfmathresult}{2}} {
	      \draw[->, color=\circularcolor] (v\i1) to (v\iplusone1);
	      \draw[->, color=\circularcolor] (v\iplusone\onewayspokelength) to (v\i\bidirectedspokelength);
      } {}
      
      \ifthenelse{\equal{\pgfmathresult}{3}} {
	      \draw[->, color=\circularcolor] (v\iplusone1) to (v\i1);
	      \draw[->, color=\circularcolor] (v\i\onewayspokelength) to (v\iplusone\bidirectedspokelength);
      } {}
}   
\end{tikzpicture}
\end{center}
\caption{The graph $G_1'$ for $l=6$. An optimum LP solution has value 1 on the blue edges and value 
$\sfrac{1}{2}$  on all other edges and hence we have $\lp = |V(G_1')|$.\label{fig:concrete_digraph_example}}
\end{figure}
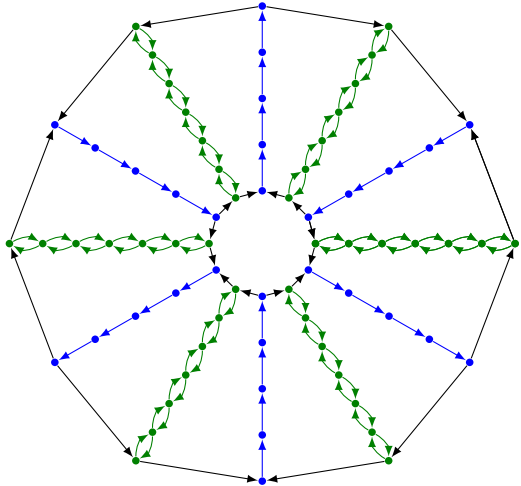

In particular, this implies that the node-weighted instances from \cite{BoyEM15} can be transformed to unweighted instances whose
integrality ratio tends to 2. 
For convenience we show these instances in Figure \ref{fig:construction_digraph_example} and Figure \ref{fig:concrete_digraph_example}.
Figure \ref{fig:construction_digraph_example} shows the general construction of the family of instances, 
Figure \ref{fig:concrete_digraph_example} a concrete example.
To obtain these instances we have replaced every vertex $v$ in the node-weighted instances
with node-weight $c_v$ by a path of length $2c_v -1$ similar to the proof of Lemma \ref{lemma:unweighted}.
So, contracting the blue paths of length $d_i$ in Figure \ref{fig:construction_digraph_example} and 
setting the node-weight of the resulting vertex to $\frac{d_i + 1}{2}$ and node-weights in $G_0$ to $\frac{1}{2}$
results in the instances from \cite{BoyEM15}.
Then, LP solutions (and tours) in the node-weighted instance correspond to LP solutions (and tours) of the same cost in the 
unweighted instance.
It seems that previously only unweighted instances with integrality ratio at most $\frac{3}{2}$ were known
(e.g. \cite{Got13}).

By splitting an arbitrary vertex into two copies $s$ and $t$, both inheriting all incident edges, this also yields a family of unweighted digraph
instances of ATSPP whose integrality ratio tends to two. We summarize:

\begin{corollary}
 The integrality ratio for unweighted digraph instances is at least two, both for \eqref{eq:subtour_lp_atsp} and \eqref{eq:subtour_lp_path}.
\end{corollary}

\newcommand{\bib}[3]{\bibitem[\protect\citeauthoryear{#1}{#2}]{#3}}

\end{document}